\documentclass{llncs}
\usepackage[utf8]{inputenc}
\usepackage{mathtools}
\usepackage{amsmath}

\DeclarePairedDelimiter{\ceil}{\lceil}{\rceil}
\DeclarePairedDelimiter{\floor}{\lfloor}{\rfloor}
\newtheorem{observation}{\textbf{Observation}}

\usepackage{graphicx} 
\usepackage{amssymb} 
\usepackage{mathtools}
\usepackage{amsmath}
\usepackage{algorithmic}
\usepackage[section]{algorithm}
\usepackage{authblk}
\usepackage{color}
\usepackage{textcomp}
\usepackage{tikz}		

\begin{document}

\title{Improved approximation algorithm for Fault-Tolerant Facility Placement}

\author{Bartosz Rybicki\thanks{Research supported by NCN 2012/07/N/ST6/03068 grant} \and Jaroslaw Byrka\\ \{bry, jby\}@ii.uni.wroc.pl}

\institute{Institute of Computer Science\\ University of Wrocław, Poland}

\maketitle

\begin{abstract}
We consider the Fault-Tolerant Facility Placement problem ($FTFP$), which is a generalization of the classical Uncapacitated Facility Location problem ($UFL$). In the $FTFP$ problem we have a set of clients $C$ and a set of facilities $F$. Each facility $i \in F$ can be opened many times. For each opening of facility $i$ we pay $f_i \geq 0$. Our goal is to connect each client $j \in C$ with $r_j \geq 1$ open facilities in a way that minimizes the total cost of open facilities and established connections. 

In a series of recent papers $FTFP$ was essentially reduced to $FTFL$ and then to $UFL$ showing it could be approximated with ratio $1.575$. In this paper we show that $FTFP$ can actually be approximated even better. We consider approximation ratio as a function of $r = min_{j \in C}~r_j$ (minimum requirement of a client). With increasing $r$ the approximation ratio of our algorithm $\lambda_r$ converges to one. Furthermore, for $r > 1$ the value of $\lambda_r$ is less than 1.463 (hardness of approximation of $UFL$). We also show a lower bound of 1.278 for the approximability of the Fault-Tolerant Facility Location problem ($FTFL$) for arbitrary $r$. Already for $r > 3$ we obtain that $FTFP$ can be approximated with ratio 1.275, showing that under standard complexity theoretic assumptions $FTFP$ is strictly better approximable than $FTFL$. 
\end{abstract}

\newpage

\section{Introduction}
In the Fault-Tolerant Facility Placement problem, we are given a set $F$ of locations where facilities may be opened (each $i \in F$ costs $f_i > 0$ and can be opened many times) and a set $C$ of clients. Each $j \in C$ has connection requirement $r_j > 0$. Our goal is to open a subset of facilities (possibly many copies of some facilities) and connect each client $j$ with $r_j$ open facilities, such that the total cost of connections and opened facilities is as small as possible. In this paper we consider the metric version of the problem where the connection costs satisfy the triangle inequality.

It is easy to see that the classical $UFL$ problem is a special case of $FTFP$ with all $r_j = 1$ . On the other hand, if no facility can be open more than once, then the problem becomes the Fault-Tolerant Facility Location problem (FTFL), in which the demands cannot exceed the number of facilities. 

Facility location problems are typically APX-hard and there exist constant factor approximation algorithms assuming metric connection costs. Shmoys, Tardos and Aardal \cite{Tardos} gave the first constant factor $3.16$-approximation algorithm based on LP-rounding. Later Chudak and Shmoys \cite{Chudak} obtained $(1 + \frac{2}{e})$-approximation by marginal-preserving randomized  rounding of facility openings, which has became standard for facility location problems. The long line of results for UFL includes a primal-dual algorithm JMS \cite{Jain}, which was then combined with a scaled version of \cite{Chudak} in a work of Byrka and Aardal \cite{Aardal}. The currently best known ration of 1.488 was obtained by Shi Li \cite{ShiLi} by further randomizing the algorithm from \cite{Aardal}. Many techniques developed for $UFL$ can be directly applied to $FTFP$ which was shown in \cite{Yan}.

First constant factor approximation algorithm for the closely related $FTFL$ problem was given by Guha, Meyerson and Munagala \cite{Meyerson}. Next Swamy and Shmoys  improved the ratio to 2.076, see~\cite{Swamy}. More recently Byrka, Srinivasan and Swamy~\cite{ftfl_1725} improved the ratio to 1.725 using dependent rounding~\cite{Aravind} and laminar clustering. Moreover it is shown in~\cite{Swamy} that JMS algorithm can be adapted to $FTFL$ with uniform requirements of clients.

$FTFP$ was first studied by Xu and Shen \cite{Xu} and next by Yan and Chrobak who first obtained a $3.16$-approximation algorithm~\cite{Yan_316}, and later improved the ratio to $1.575$~\cite{Yan}.

\subsection{Our contribution}

We extend the work of Yan and Chrobak~\cite{Yan} and propose an algorithm with approximation ratio being a decreasing function of the minimal requirement $r = min_{j \in C}~r_j$. Our solution benefits from requirements of clients being bigger than one. Instead of considering a client $j \in C$ as $r_j$ distinct clients with unit demand we raven on this multiplicity and use Poisson distribution to estimate the expected number of useful facilities which will be open in a set of a particular volume. We consider both cases: uniform and non-uniform requirements of clients, and obtain the following approximation ratios: 

\begin{center}
  \begin{tabular}{ c | c | c | c | c | c | c | c | c | c | c }
    $r$ & 1 & 2 & 3 & 4 & 5 & 6 & 7 & 8 & 9 & 10 \\ \hline
    non-uniform & 1.515 & 1.439 & 1.338 & 1.275 & 1.234 & 1.207 & 1.187 & 1.171 & 1.159 & 1.149 \\ \hline
    uniform & 1.488 & 1.410 & 1.329 & 1.272 & 1.234 & 1.207 & 1.187 & 1.171 & 1.159 & 1.149 \\
  \end{tabular}
\end{center}

We also prove a lower bound of 1.278 on the approximability of Fault-Tolerant Facility Location (where at most one facility may be opened in each location) for arbitrarily large $r > 1$. 

\begin{observation}
 Lower bound for $FTFL$, of value $1.278$, is bigger than $\lambda_r$ for $r \geq 4$. Moreover for $r \geq 2$ $FTFP$ is easier than $UFL$
\end{observation}

Note that $\lambda_r$ for $r = 4$ (in both uniform and non-uniform case) is bounded by $1.275$, which is smaller than our lower bound for $FTFL$.

\section{The LP formulation}
\label{the_lp_section}

Consider the following standard LP relaxation of $FTFP$.

\begin{eqnarray}
\label{lp_ufl:goal}
  min \sum_{i \in F}{\sum_{j \in C} {c_{ij}x_{ij}}} & + & \sum_{i \in F} y_if_i\\
 \label{lp_ufl:r_satisfy}
  \sum_{i \in F} x_{ij}&\geq& r_j ~~~~\forall_{j \in C}\\
 \label{lp_ufl:open_enough}
   y_i - x_{ij}&\geq& 0 ~~~~~\forall_{i \in F, j \in C} \\
 \label{lp_ufl:non_negative}
   x_{ij}, y_{i} &\geq& 0 ~~~~~\forall_{i \in F, j \in C}
\end{eqnarray}

An optimal solution of the above LP is denoted by a pair $(x^*, y^*)$. Using these variables we express the total facility cost as $F^* = \sum_{i \in F}f_iy_i^*$ and the connection cost of each client $j \in C$ as $C_j^* = \sum_{i \in F}c_{ij}x_{ij}^*$. Summing over all clients gives the total connection cost $C^* = \sum_{j \in C}C_j^*$ of the LP solution. The cost of $(x^*, y^*)$ denoted by $LP^* = F^* + C^*$ is a lower bound on the cost of an optimal integral solution denoted by $OPT$.

We say that a solution is complete if for each client $j \in C$ and each facility $i \in F$ we have $x_{ij}^* \in \{0, y_i^*\}$. Detailed description of a technique called \textit{facility splitting}, which yields complete solutions, can be found in \cite{Sviridenko}. The splitting algorithm takes as input a solution of the LP and outputs a complete solution of the same cost to a larger, but equivalent instance of the problem. For clarity of a presentation, throughout the paper, we simply assume that all fractional solutions are complete.

\begin{definition}
 The volume of a set $F' \subseteq F$, denoted by $vol(F')$ is the sum of facility openings in this set, i.e., $vol(F') = \sum_{i \in F'} y_i$.
\end{definition}

One of the problems with input instances is possibly non-polynomial demand of some clients. In \cite{Yan} we can find an elegant reduction of such instance to instances with requirements bounded by $|F|$. In Section \ref{r_vs_apx} we give an algorithm which generalizes this reduction. Our algorithm also reduces the input instance to an instance with polynomial demands of clients, but we also care not to reduce the requirements of clients too much.

\section{Algorithm for $FTFP$}
The following algorithm $A(\gamma)$ is parametrized by a real constant $\gamma \in (1,3)$.
\begin{algorithm}
  \caption{$A(\gamma)$}
\begin{algorithmic}[1]
 \STATE formulate and solve the LP (\ref{lp_ufl:goal})-(\ref{lp_ufl:non_negative}), get an optimal solution $(x^*, y^*)$;
 \STATE scale up facility opening by $\gamma$, then recompute values of $x_{ij}$ to obtain a minimum cost solution $(\bar{x}, \bar{y})$;\\
 \STATE compute clustering for all clients;
 \STATE round facility opening variables using dependent rounding;
 \STATE connect each client $j$ with $r_j$ closest open facilities;
\end{algorithmic}
\end{algorithm}
Our final $Algorithm~1$ is as follows: run algorithm $A(\gamma_l)$ for each choice of $\gamma_l = 1 + 2 \cdot \frac{n - l}{n}$, where $l = 1, 2, \dots n-1$. Select the best of the obtained solutions. Note that $n-1$ is the number of different values of $\gamma$, each of them we use as a parameter of algorithm $A(\gamma)$. In the computation of approximation ratios we use $n$ equal 1000, but we will describe our results for a general n.

Scaling facility opening is an idea from \cite{Aardal}, it decreases average connection cost of each client, but increases total cost of opening facilities. In $FTFP$ we can open more than one facility in one location, so scaling does not cause problems with opening more than one facility in one place. The version of clustering which we use is very close to the one described in \cite{Chudak}. To round facility opening variables we use the randomized algorithm from \cite{Aravind}, called dependent rounding. Each step of the algorithm $A(\gamma)$ is carefully described in the following sections.

\subsection{Scaling}

Let $F_j$ denote the set of facilities with a positive flow from a client $j \in C$, i.e., facilities $i$ with $x_{ij}^* > 0$ in the optimal LP solution. 

Let $\gamma_l > 1$. Suppose that all facilities are sorted in an order of non-decreasing distances from a client $j \in C$. Scaling all $y^*$ variables by $\gamma_l$ divides the set of facilities $F_j$ into two disjoint subsets: the set of close facilities of a client $j$, denoted by $F_j^{C_l}$, such that $vol(F_j^{C_l}) = r_j$;  and the distant facilities, denoted by $F_j^{D_l} = F_j \setminus F_j^{C_l}$, note that $vol(F_j^{D_l}) = r_j(\gamma_l - 1)$. Certainly for each $i_1 \in F_j^{C_l}$ and $i_2 \in F_j^{D_l}$ we have $c_{i_{1}j} \leq c_{i_{2}j}$.

By $D_{av}^{C_l}(j), D_{av}^{D_l}(j)$ and $D_{av}(j)$ we denote the average distances to close, distant and all facilities in set $F_j$, respectively. More formally: $$D_{av}^{C_l}(j) = \frac{\sum_{i \in F_j^{C_l}} c_{ij} \bar{x}_{ij}}{vol(F_j^{C_l})}, \; \; \; D_{av}^{D_l}(j) = \frac{\sum_{i \in F_j^{D_l}} c_{ij} \bar{x}_{ij}}{vol(F_j^{D_l})}$$ By $D_{max}^l(j)$ we denote the maximal distance to a facility in $F_j^{C_l}$, and by $c_l(j)$ we denote the average distance to $F_j^l = F_j^{C_l} \setminus F_j^{C_{l-1}}$ for $n > l \geq 1$, $F_j^n = F_j \setminus F_j^{C_{n-1}}$ and $F_j^0 = \emptyset$. (see Fig.~\ref{partition_of_F_j})

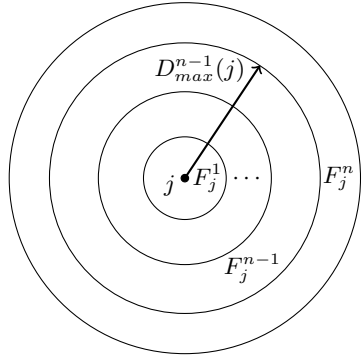
\begin{figure}
  \begin{tikzpicture}
      \draw[fill=black] (6,.5) circle (0.05);
      \path (5.8, 0.4) node {$j$};
      
      \draw (6,.5) circle (.55);
      \path (6.3, 0.5) node {$F_j^1$};
      
      \draw (6,.5) circle (1.15);
      \path (6.85, 0.5) node {$\cdots$};
      
      \draw (6,.5) circle (1.8);
      \path (6.9, -.7) node {$F_j^{n-1}$};
      \draw[thick, ->] (6, 0.5) -- (7, 2);
      \path (6.2, 1.95) node {$D_{max}^{n-1}(j)$};
      
      \draw (6,.5) circle (2.335);
      \path (8.05, 0.5) node {$F_j^n$};
      
      \draw (0.5,.5) circle (0);
  \end{tikzpicture}
  \caption{Figure shows partition of facilities in set $F_j$.}
  \label{partition_of_F_j}
\end{figure}
 
\subsection{Clustering}

\begin{definition}
 The radius of a set $A$ for a client $j$, where $A \subseteq F$ and $j \in C$, is $max_{i \in A}~c_{ij}$. Assume that $vol(A) \geq r$. By $B(j, A, r)$ we denote the subset of $A$ of volume $r$ which has the smallest radius.
\end{definition}

Each client $j \in C$ initially has a cluster proposition $CP(j) = B(j, F, r_j) = F_{j}^C$, whose radius is $q_j = D_{max}^{C}(j)$. In the following algorithm the cluster proposition of a client $j$ changes, but the radius never increases.
\begin{algorithm}
  \caption{Clustering}
\begin{algorithmic}[1]
 \FORALL{$j \in C$}
  \STATE $q_j := D_{max}^{C}(j)$
 \ENDFOR
 \WHILE{there is a client with positive requirement}
  \STATE select a client $j \in C$ with $r_j > 0$ that minimizes $q_j$ and set $r_j := 0$
  \FORALL{$j' \in N(j) = \{j'' \in C~ |~ CP(j) \cap CP(j'') \neq \emptyset \wedge r_{j''} > 0\}$}
   \STATE $r_{j'} := max(0, r_{j'} - \ceil{vol(CP(j) \cap CP(j'))})$
   \STATE $CP(j') := B(j', CP(j') \setminus CP(j), r_{j'})$
  \ENDFOR
  \STATE create $C(j) = \{j\} \cup N(j) \cup CP(j)$ and call $j$ the center of cluster $C(j)$;
 \ENDWHILE
\end{algorithmic}
\end{algorithm}

The above described procedure is a variant of the method described in \cite{Chudak}. It is well known that output of the procedure has two important properties. First: each facility is clustered by at most one client. Second: the distance from a client to any of his cluster centers is not too big.

\begin{lemma}
\label{3hop_bound}
Distance from any client $j \in C$ to any close facility of $j' \in C$ such that $j \in C(j')$ is bounded by $3 \cdot D_{max}^C(j)$.
\end{lemma}

\begin{proof}
Suppose that $j' \in C$ and $j \in C_{C}(j') = C(j') \cap C$. From the fact that $j \in C_{C}(j')$ follows that $q_{j'} \leq q_{j}$, which is equivalent with $D_{max}^C(j') \leq D_{max}^C(j)$. The definition of $CP(\cdot)$ assures that the distance from $j$ ($j'$) to any facility in $CP(j)$ ($CP(j')$) can be bounded by $D_{max}^C(j)$ $(D_{max}^C(j'))$. Consider $i' \in CP(j) \cap CP(j')$ and any $i \in CP(j')$. Distance from $j$ to $i$ is $c_{i'j} + c_{i'j'} + c_{ij'} \leq D_{max}^C(j) + 2 \cdot D_{max}^C(j') \leq 3 \cdot D_{max}^C(j)$.
\qed
\end{proof}

\subsection{Facility opening}

A randomized procedure deciding whether a particular facility should be open or not transforms the fractional $\bar{y}$ into a random integral $\hat{y}$. 
We would like the procedure to have the following properties: 
\begin{enumerate}
 \item \emph{Marginal distribution}:$~Pr[\hat{y}_i = 1] = \bar{y}_i$
 \item \emph{Sum-preservation}:$~\sum_{i \in C_F(j)} \hat{y}_i \in \{\floor{vol(C_F(j))}, \ceil{vol(C_F(j))}\}$
 \item \emph{Negative correlation}:$~\displaystyle \forall S \subseteq C_F(j) \forall b \in \{0, 1\} Pr[\bigwedge_{i \in S}(\hat{y}_i = b)] \leq \prod_{i \in S} Pr[\hat{y}_i = b]$
\end{enumerate}

One method which gives an output with the above properties is the dependent rounding (DR) from \cite{Aravind}. 
Each cluster can have many facilities open fractionally. We first apply DR to each $C_F(j) = C(j) \cap F$, where $j$ is the center of a cluster. 
Then the remaining fractional facility openings are rounded by DR in an arbitrary order.

\section{Analysis}

To bound the expected connection cost of an algorithm $A(\gamma)$, we need to first analyse the number of facilities which will be opened in a set of a particular volume. Suppose that facilities are opened independently and that in the limit case all facilities are opened very little in the fractional solution, then the number of eventually open facilities from a set has the Poisson distribution. By the negative correlation this distribution can be used to derive the following lower bound on the number of useful opened facilities from the considered set.

\begin{observation} 
The expected number of possible connections with set $A$ of volume $\varLambda = vol(A)$, when the requirement is k, is $h(\varLambda, k) \geq \sum_{i = 1}^{k-1}i P_{\varLambda}(X = i) + k P_{\varLambda}(X \geq k)$. Where $P_{\varLambda}(X = i) = \frac{\varLambda^i e^{-\varLambda}}{i!}$ is the probability of opening exactly $i$ facilities in a set of volume $\varLambda$, if opened independently (Poisson distribution).
\end{observation}

\begin{lemma}
\label{alg_A_connection_cost}
Suppose that $\gamma=\gamma_k$. Consider a client $j \in C$ which is not a center of any cluster. The expected connection cost of client $j$ is at most
$$E[C_j] \leq \sum_{l = 1}^{n-1} c_l(j) \cdot \frac{e^{k,l}_{1}(j)}{r_j} + \frac{e^{k}_{3}(j)}{r_j} \cdot 3 D_{max}^k(j)$$
Where $e_1^{k, l}(j)$ is expected number of open facilities in set $F_j^l$, in which opening of each facility is scaled by $\gamma_k$; $e_3^{k}(j)$ is $r_j$ decreased by expected number of open facilities in set $F_j$, in which opening of each facility is scaled by $\gamma_k$ (or zero if number of open facilities in $F_j$ is bigger than $r_j$).
\end{lemma}

\begin{proof}
 The value of $e_1^{k, l}(j)$ is the expected number of open facilities in the set $F_j^l = F_j^{C_l} \setminus F_j^{C_{l-1}}$, when all fractional openings of facilities are scaled up by $\gamma_k$. Connection cost of a client $j$ with an open facility in this set is $c_l(j)$. The expected number of connections which $j$ has to establish with close facilities of his cluster centers is $e_3^{k}(j)$ - his requirement reduced by the number of facilities opened in $F_j$. Lemma (\ref{3hop_bound}) bounds the distance to close facilities of cluster centers of $j$.
 \qed
\end{proof}

\subsection{Factor revealing LP}
\label{lp_factor_revel}
Consider running an algorithm $A(\gamma_l)$ for $\gamma_l = 1 + 2 \cdot \frac{n - l}{n}$ where $l = 1, 2, \dots n-1$. Observe that the following linear program, called FRLP, covers all that executions. Value of the objective function is an upper bound on the approximation ratio of the best of the obtained solutions.
\begin{eqnarray}
  \label{total_lp:max}
  max~\lambda_r && \\
  \label{total_lp:t_upperbound}
  \gamma_k f + \sum_{l = 1}^{n-1} c_l \cdot \frac{e_1^{k, l}}{r} + \frac{e_3^k}{r} \cdot 3 \cdot c_{l+1} \geq \lambda_r &&  \forall_{k < n}\\
  \label{total_lp:c_constraint}
  \sum_{l = 1}^{n}vol(F^l) \cdot c_l = c && \\
  \label{total_lp:c_order}
  0 \leq c_i \leq c_{i+1} \leq 1&&\forall_{i < n}~~~~~\\
  \label{total_lp:opt_sol}
  f + c = 1 && \\
  f, c\geq 0 &&
\end{eqnarray}

The above $LP$ encodes the cost of solutions obtained in executions of an algorithm $A$ for different values of the scaling parameter $\gamma_k$ for $k = 1, 2, \dots n-1$. Adversary has the freedom to choose the distances from client $j$ to groups of facilities and the relation between values of $f$ and $c$ in the optimal solution, which have to sum up to one, and (both) be non-negative. We consider all facilities in the order of a non-decreasing distance from the client $j$, so the average distances to consecutive groups of facilities have to be non-decreasing, see constraint (\ref{total_lp:c_order}). We divide facilities into sets $F_j^l$, for $1 \leq l \leq n$. In each set $F_j^l$ the adversary may choose the distance from client $j$ to the open facility in $F_j^l$, which is the worst for our algorithm and equals $c_l(j)$. Equality (\ref{total_lp:c_constraint}) shows that the sum of average distances, each weighted by the volume of facilities at such distance, has to sum up to the total connection cost in 
the optimal solution.  
The crucial inequality (\ref{total_lp:t_upperbound}) encodes the expected cost of an algorithm $A(\gamma_k)$ and it is used as an upper bound for the approximation ratio. Client in inequality (\ref{total_lp:t_upperbound}) is a client with minimum requirement $r$,
$e_1^{k,l} = h(\gamma_k \cdot vol(F^{C_l}, r)) - h(\gamma_k \cdot vol(F^{C_{l-1}}), r)$ is expected number of open facilities in set $F^l$ and $e_3^{k} = r - h(\gamma_k \cdot r, r)$. Correctness of this inequality follows from Lemmas (\ref{3hop_bound}), (\ref{alg_A_connection_cost}) and $D_{max}^{l} \leq c_{l+1}$.
If $r = 1$ then instead of $Algorithm~1$ we use method from \cite{Yan}. To improve the approximation ratio from $1.575$ to $1.52$ we run the algorithm from~\cite{Yan} for a number of values of the scaling parameter $\gamma_l = 1 + 2 \cdot \frac{n-l}{n}$, where $l = 1, 2,\ldots n-1$. It can be analyzed by FRLP. The computed values of $\lambda_r$, for $r = 1, 2, \dots 10$, are in the following table:

\begin{center}
  \begin{tabular}{ c | c | c | c | c | c | c | c | c | c | c }
    $r$ & 1 & 2 & 3 & 4 & 5 & 6 & 7 & 8 & 9 & 10 \\ \hline
    $\lambda_r$ & 1.515 & 1.439 & 1.338 & 1.275 & 1.234 & 1.207 & 1.187 & 1.171 & 1.159 & 1.149 \\
  \end{tabular}
\end{center}

\begin{figure}
  \centering
  \includegraphics[height=60mm]{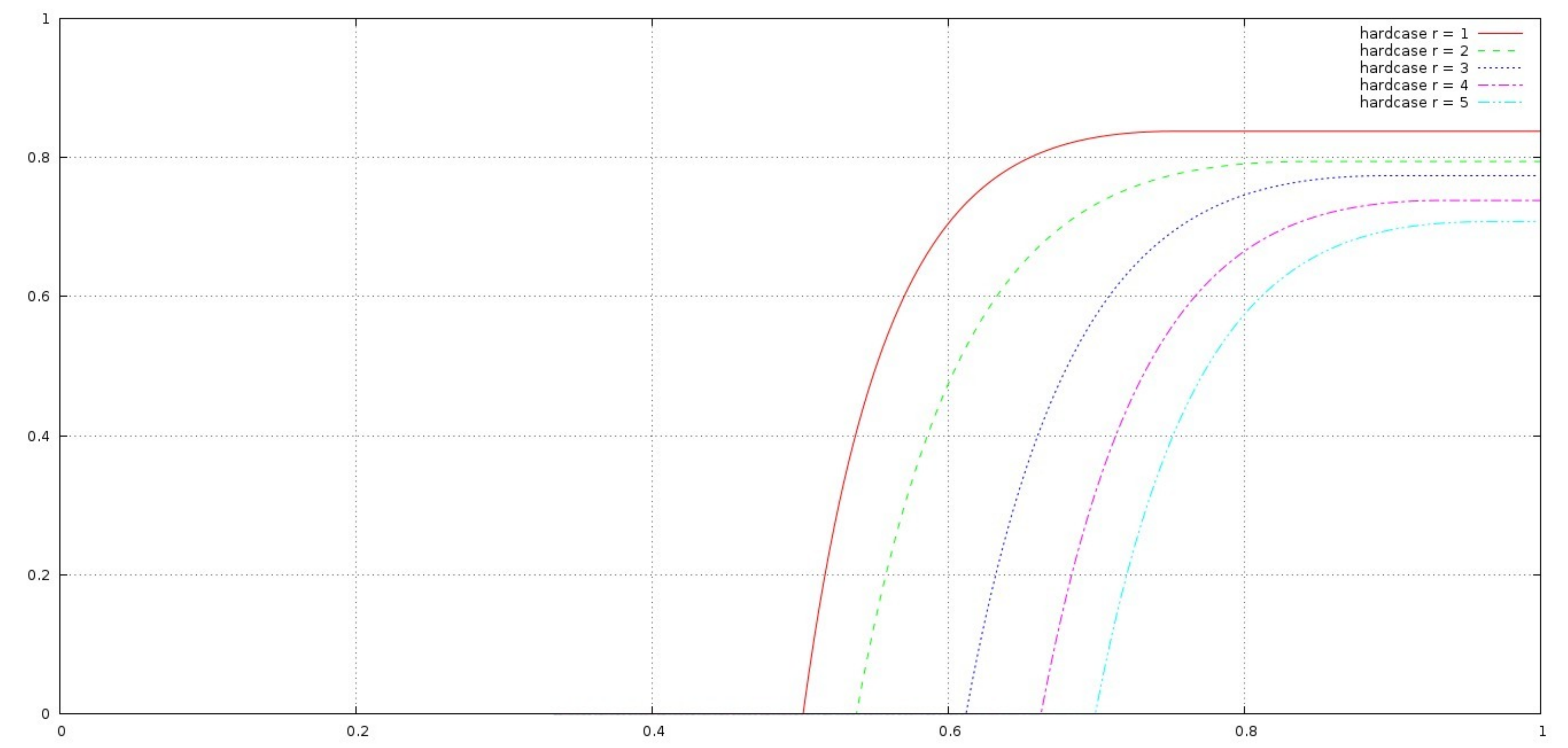}
 \caption{The profiles of distances in tight instances for $Algorithm~1$ for $FTFP$ (in a general, non-uniform case) for $1 \leq r \leq 5$, extracted from the $FRLP$ solutions. The x-axis encodes the volume of a set of facilities closest to a client and the y-axis is the distance to the farthest facility in this set.}
  \label{fig:hardcase}
\end{figure}

\subsection{Uniform requirement}
As it was shown in \cite{Swamy} the JMS algorithm can be modified to work with $FTFL$ with uniform requirements of clients, and the approximation ratio remains the same. In consequence it also works for $FTFP$ with uniform requirements of clients. We can add one more constraint $1.11 f + 1.78 c \geq \lambda_r$ to the $FRLP$ in Section~\ref{lp_factor_revel} which encodes that we additionally run the (modified) JMS algorithm\footnote{An algorithm for UFL is called (a,b)-approximation if the cost of returned solution is upper bounded by $a \cdot F^* + b \cdot C^*$, where $F^*$ and $C^*$ are, respectively, the costs of establishing connections and opening facilities in an optimal solution}. Such FRLP for $r = 1$ is a dual of the LP from \cite{ShiLi}, probabilities of particular algorithms in Shi Li paper are dual values of constraints in FRLP. As you can see in the following table, adding the JMS algorithm makes difference only for small values of $r$.

\begin{center}
  \begin{tabular}{ c | c | c | c | c | c | c | c | c | c | c }
    $r$ & 1 & 2 & 3 & 4 & 5 & 6 & 7 & 8 & 9 & 10 \\ \hline
    non-uniform & 1.515 & 1.439 & 1.338 & 1.275 & 1.234 & 1.207 & 1.187 & 1.171 & 1.159 & 1.149 \\ \hline
    uniform & 1.488 & 1.410 & 1.329 & 1.272 & 1.234 & 1.207 & 1.187 & 1.171 & 1.159 & 1.149 \\
  \end{tabular}
\end{center}

\section{Factor $\lambda_r$ is a decreasing function of $r$}
\label{r_vs_apx}

\begin{lemma}
\label{technical_lemma}
 Function $f(r) = \frac{r - h((1+\epsilon)r, r)}{r}$ converges to 0 when $r \mapsto \infty$.
\end{lemma}

\begin{proof}
We show that 
$$\lim_{r \mapsto \infty} \frac{r - h((1+\epsilon)r, r)}{r} = 0$$
\begin{eqnarray*}
\frac{r - h((1+\epsilon)r, r)}{r} &=& 
\frac{r - (\sum_{i = 1}^{r-1}iP_{(1+\epsilon)r}(X = i) + r P_{(1+\epsilon)r}(X \geq i))}{r} \nonumber \\
& = &\frac{r - (\sum_{i = 1}^{r-1}iP_{(1+\epsilon)r}(X = i) + r (1 - \sum_{i = 1}^{r-1}P_{(1+\epsilon)r}(X = i)))}{r} \nonumber \\
& = &\frac{\sum_{i = 1}^{r-1} (r-i) P_{(1+\epsilon)r}(X = i) }{r} \nonumber \\
& = &\sum_{i = 1}^{r-1} \frac{(r-i) (1+\epsilon)^i \cdot r^i }{r\cdot i! \cdot e^{(1+\epsilon)r}} \leq \frac{(1+\epsilon)^{r-1} \cdot r^r}{(r-1)!e^{(1+\epsilon)r}} \nonumber
\end{eqnarray*}

It remains to verify that $$\forall_{\epsilon > 0} \lim_{r \mapsto \infty} \frac{(1+\epsilon)^{r-1} \cdot r^r}{(r-1)!e^{(1+\epsilon)r}} = 0,$$ 
which we did using wolfram alpha (http://www.wolframalpha.com/).
\end{proof}

\begin{theorem}
\label{apx_r_relation}
 For each choice of $\epsilon > 0$ there exists $r_0$ such that for each $r \geq r_0$ inequality $\lambda_r~\leq~1~+~\epsilon$ holds.
\end{theorem}
The above theorem easy follows from the following lemmas, because the approximation ratio of $Algorithm~1$ is always upper bounded by the approximation ratio of $A(\gamma)$ for each choice of $\gamma$.

Consider an instance $\mathcal{I}$ and a client $j \in C$. Lemma (\ref{alg_A_connection_cost}) implies that the expected connection cost of $j$ in algorithm $A(\gamma)$ is $E[C_j] \leq \sum_{l = 1}^{n} c_l(j) \cdot \frac{e_1^{k,l}(j)}{r_j} + \frac{e_3^{k}(j)}{r_j} \cdot 3 D_{max}(j)$. Note that the inequality $\sum_{l = 1}^{n} c_l(j) \cdot \frac{e_1^{k,l}(j)}{r_j} \leq C_j^*$ holds, because in the solution $(x^*, y^*)$ client $j$ fractionally uses the same facilities, but with smaller opening values. Therefore, in the expectation he pays not less for connection than in our scaled up solution. Notice that, for a particular choice of $\gamma = 1 + \epsilon$, the value of the expression $3 (1 + \frac{1}{\epsilon})$ is a constant. From \cite{Aravind} we know that the following inequality holds. $$\frac{e_3^{k}(j)}{r_j} \cdot 3 D_{max}(j) \leq f(r) \cdot 3 (1 + \frac{1}{\epsilon}) C_j^*$$

\begin{observation}
 \label{nice_observation}
 For $\gamma = 1 + \epsilon$ the approximation factor for connection cost of the solution produced by $A(\gamma)$ depends only on $f(r)$, where $r$ is the minimum requirement in the considered instance.
\end{observation}

Li Yan showed \cite{liyan_phd} a result similar to the below lemma, but the result is weaker: he shows that the limit is 1 only for a fixed number of facilities.

\begin{lemma}
\label{1_epsilon}
 For each $\epsilon > 0, \gamma = 1 + \epsilon$, there exists $r_0$ such that for each $r \geq r_0$, approximation ratio of an algorithm $A(\gamma)$ is bounded by $1 + \epsilon$.
\end{lemma}

\begin{proof}
 Lemma~\ref{technical_lemma} and Observation~\ref{nice_observation} imply that for each choice of $\epsilon$ there exists $r_0$ such that for each instance with minimum requirement $r \geq r_0$ approximation ratio of $A(1+\epsilon)$ is bounded by $1 + \epsilon$.
 \qed
\end{proof}

\subsection{Dealing with large requirements $r_j$}
Yan and Chrobak proved the following theorem

\begin{theorem} (from \cite{Yan})
\label{requirements_reduction}
 Suppose that there is a polynomial-time algorithm $\mathcal{A}$ that, for any instance of $FTFP$ with maximum demand bounded by $|F|$, computes an integral solution that approximates the fractional optimum of this instance within factor $\rho > 1$. Then there is a $\rho$-approximation algorithm $\mathcal{A'}$ for $FTFP$.
\end{theorem}

The main result of this section is an extension of Theorem \ref{requirements_reduction} which exploits our Theorem \ref{apx_r_relation}. Consider an instance $\mathcal{I}$ for which the approximation ratio of $Algorithm~1$ is almost one, see Theorem \ref{apx_r_relation}. As it was mentioned in Section \ref{the_lp_section} we can assume that the optimal solution $(x^*, y^*)$ to the LP (\ref{lp_ufl:goal}) - (\ref{lp_ufl:non_negative}) for an instance $\mathcal{I}$ is complete, so for each $i \in F$ and $j \in C$ we have $x_{ij}^* \in \{0, y_i^*\}$ . From optimality of this solution, we can assume that $\sum_{i \in F}x_{ij}^* = r_j$ for all $j \in C$. We split solution $(x^*, y^*)$ into two parts $(x^*, y^*) = (\hat{x}, \hat{y}) + (\dot x, \dot y)$, where 
$$\hat{y}_i = max\{\floor{y_i^* - \bar r}, 0\},~~~\hat{x}_{ij} = max\{\floor{x_{ij}^* - \bar r}, 0\}~~\forall j \in C, i \in F$$
$$\dot y_i = y_i^* - \hat{y}_i,~~\dot x_{ij} = x_{ij}^* - \hat{x}_{ij}~~\forall j \in C, i \in F$$
where $1 \leq \bar r \leq min_{j \in C}~r_j$. Now we will construct two instances $\mathcal{\dot I}$ and $\hat{\mathcal{I}}$ of $FTFP$ with the same parameters as $\mathcal{I}$, except requirements. Demand of each client $j$ is $\hat{r}_j = \sum_{i \in F} \hat{x}_{ij}$ in the instance $\hat{\mathcal{I}}$ and $\dot r_j = \sum_{i \in F} \dot x_{ij}$ in $\mathcal{\dot I}$.

\begin{lemma}
 (i) $(\hat{x}, \hat{y})$ is a feasible integral solution to instance $\hat{\mathcal{I}}$ \\
 (ii) $(\dot x, \dot y)$ is a feasible fractional solution to instance $\mathcal{\dot I}$ \\
 (iii) $(\hat{x}, \hat{y})$ and $(\dot x, \dot y)$ are optimal solutions to $\mathcal{\hat I}$ and $\mathcal{\dot I}$ \\
 (iv) $\forall_{j \in C}~ (\bar r + 1) \cdot |F| \geq \dot r_j \geq \bar r$
\end{lemma}

\begin{proof}
 (i) For a feasibility of $(\hat{x}, \hat{y})$, we need to show that all constraints of LP (\ref{lp_ufl:goal}) - (\ref{lp_ufl:non_negative}) are satisfied. For each $j \in C$ we have that $\hat{r}_j = \sum_{i \in F} \hat{x}_{ij}$, so (\ref{lp_ufl:r_satisfy}) holds. Solution $(x^*, y^*)$ is complete, so $x_{ij}^* \in \{0, y_i^*\}$. If $x_{ij}^* = 0$ then $\bar{x}_{ij} = 0 \leq \bar{y}_i$, otherwise $x_{ij}^* = {y}_i^* > 0$ in that case we have that $\hat{x}_{ij} = \hat{y}_{i}$. In consequence constraint (\ref{lp_ufl:open_enough}) is satisfied.
 
 (ii) In the case of $(\dot x, \dot y)$ also all inequalities hold. Constraint (\ref{lp_ufl:r_satisfy}) is satisfied, because $\dot r_j = \sum_{i \in F} \dot x_{ij}$. Note that both $\dot x_{ij}$ and $\dot y_i$ are non-negative. We need to show that $\dot x_{ij} \leq \dot y_i$ which is equivalent with $y_i^* - max\{\floor{y_i^* - \bar r}, 0\} \geq x_{ij}^* - max\{\floor{x_{ij}^* - \bar r}, 0\}$. If $x_{ij}^* = 0$ then we have $y_i^* \geq max\{\floor{y_i^* - \bar r}, 0\}$ which holds. In the other case we have $x_{ij}^* = y_i^* > 0$. With that assumption we trivially obtain the following equality $y_i^* - max\{\floor{y_i^* - \bar r}, 0\} = x_{ij}^* - max\{\floor{x_{ij}^* - \bar r}, 0\}$.
 
 (iii) Suppose that at least one of $(\hat{x}, \hat{y})$ and $(\dot x, \dot y)$ is not an optimal solution to $\mathcal{\hat I}$ and $\mathcal{\dot I}$, respectively. In that situation we are able to obtain solution to instance $\mathcal{I}$ with a smaller cost than $cost(x^*, y^*)$, which is a contradiction.
 
 (iv) To prove $\dot r_j \leq (\bar r + 1) \cdot |F|$ we have to show that the following inequality holds,  where $F' = \{i \in F ~|~ x_{ij}^* \geq \bar r + 1\}$. $$r_j - \sum_{i \in F'} (x_{ij}^* - (\bar r + 1)) \leq (\bar r + 1)|F| \iff \sum_{i \in F \setminus F'} x_{ij}^* \leq (\bar r + 1) |F \setminus F'|$$
 To finish the proof of the lemma we should prove the following inequalities
 $$r_j - \sum_{i \in F} max\{\floor{x_{ij}^* - \bar r}, 0\} \geq r_j - \sum_{i \in F} max\{x_{ij}^* - \bar r, 0\} \geq \bar r$$
 Let $F' = \{i \in F | x_{ij}^* > \bar r\}$. Using $F'$ we can rewrite the above inequality as $r_j - \sum_{i \in F'} (x_{ij}^* - \bar r) \geq \bar r$. Consider two cases: $|F'| = 0$ and $|F'| \geq 1$. The first is trivial because $r_j \geq \bar r$ holds. In the second case $r_j - \sum_{i \in F'} (x_{ij}^* - \bar r) \geq r_j + \bar r - \sum_{i \in F}x_{ij}^* \geq \bar r$, which trivially holds, because $r_j = \sum_{i \in F}x_{ij}^*$.
\qed
\end{proof}

\begin{corollary}
 For an instance $\mathcal{I}$, with requirement $r$, for which the approximation ratio of $Algorithm~1$ is $\lambda_r$, we can obtain two other instances: integral $\mathcal{\hat I}$ and fractional $\mathcal{\dot I}$. Instance $\mathcal{\dot I}$ has polynomial demands and a minimum requirement $1 \leq \bar r \leq min_{j \in C}~r_j$. The approximation ratio $\lambda$ and a running time of $Algorithm~1$ depends on value of $\bar r$, which can be arbitrarily selected from $[1,r] $. Sum of the integral solution $\mathcal{S}$ and the optimal integral solution for $\mathcal{\hat I}$, is a feasible integral solution for $\mathcal{I}$ with the approximation ratio $\lambda$.
\end{corollary}

\section{Integrality gap of the LP}
\label{int_gap_section}
We will present examples of instances of $FTFP$, with an uniform requirement r of all clients, which has a quite big integrality gap. Let $I(n, l, f_c)$ denote an instance of the $FTFP$ defined as follows. There is a set of facilities $F = \{1, 2, \dots, n \cdot r\}$ (each of $n$ facilities has $r$ copies to avoid opening any of them more than once, in a consequence our results for an integrality gap holds also for $FTFL$) The set of clients $C = \{j \subset F : |j| = m \cdot r = l, m \in \mathbb{N} \setminus \{0\}\}$ constructed from subsets of $F$ such that each of $m$ selected facilities is in $r$ copies. A connection cost $c_{ij}$ is equal $\frac{1}{|C|}$ if $i \in j$ and $\frac{3}{|C|}$ otherwise. Instances described above have the following properties:

\begin{itemize}
\renewcommand{\labelitemi}{$\bullet$}
 \item The optimal fractional solution to the linear relaxation of $FTFP$ for these instances is: open $\frac{r}{l}$ of each facility, connect each client $j$ to all facilities whose distance is $\frac{1}{|C|}$
 \item The cost of the optimal fractional solution to the $LP$-relaxation of $FTFP$ on these instances is $z_{LP}(I(n, l, f_c)) = r(1 + \frac{f_c}{l})$ 
 \item The cost of an integral solution to the $FTFP$ depends only on number of open facilities and requirement of clients.
\end{itemize}

Now we will analyse the cost of $z(I(n, l, f_c))$, optimal solutions to instances of the form $I(n, l, f_c)$. We consider the instance obtained by setting $l$ and $f_c$ to some constant values and $n = 1, 2, 3, \dots$. We consider all solutions by modifying the parameter $\alpha \in (0, 1]$ which is a fraction of open facilities. The cost of the integral solution for such instances can be described by the following expression: $$\lim_{n \mapsto \infty}z(I(n, l, f_c)) = \alpha \cdot f_c + r + \sum_{i = 0}^{r-1} \binom{l}{i} \cdot \alpha^i \cdot (1 - \alpha)^{l-i} \cdot (r - i) \cdot 2$$ We open $\ceil{\alpha \cdot n}$ facilities of total cost $\alpha \cdot f_c$. Each client has a requirement $r$ and $\binom{l}{i} \cdot \alpha^i \cdot (1 - \alpha)^{l-i}$ represents the following event: $i$ facilities at distance $\frac{1}{|C|}$ are open. The remaining $r - i$  facilities which $j$ uses are in a distance $\frac{3}{|C|}$ from $j$, so he has to pay extra for a connection. The integrality gap for a particular 
values of requirement $r$, which we are able to find, are presented in the below table.

\begin{center}
  \begin{tabular}{ c | c | c | c | c | c | c | c | c | c | c }
    $r$ & 1 & 2 & 3 & 4 & 5 & 6 & 7 & 8 & 9 & 10\\ \hline
    IG & 1.4627 & 1.3268 & 1.2655 & 1.2289 & 1.1709 & 1.1179 & 1.076 & 1.0466 & 1.0268 & 1.0146 \\
  \end{tabular}
\end{center}

The following table presents values of parameters used to obtain our integrality gap instances. All values of the variables in the below table were obtained in experimental way.
\begin{center}
  \begin{tabular}{ c | c | c | c | c }
    $r$ & $IG$ & $\alpha$ & $f$ & $l$ \\ \hline
    1  & 1.46272 & 0.001462 & 463.495 & 1000 \\ \hline
    2  & 1.32689 & 0.002654 & 514.615 & 1000 \\ \hline
    3  & 1.26557 & 0.003796 & 539.050 & 1000 \\ \hline
    4  & 1.22895 & 0.004916 & 554.235 & 1000 \\ \hline 
    5  & 1.17098 & 0.005855 & 526.635 & 1000 \\ \hline
    6  & 1.11795 & 0.006707 & 452.550 & 1000 \\ \hline
    7  & 1.07640 & 0.007535 & 356.055 & 1000 \\ \hline
    8  & 1.04669 & 0.008374 & 257.735 & 1000 \\ \hline
    9  & 1.02687 & 0.009242 & 172.485 & 1000 \\ \hline
    10 & 1.01460 & 0.010146 & 107.175 & 1000 \\
  \end{tabular}
\end{center}

\section{Lower bound for FTFL}
We give a reduction from the Set Cover problem. Consider an instance of Set Cover defined as $X = \{x_1, x_2, \dots x_n\}$, and $S = \{S_1, S_2, \dots S_m\}$ such that $S_i \subseteq X$ for each $i \in \{1, 2, \dots m\}$. We would like to find a cover $C \subseteq S$ such that $|C| = k$ is minimized. In our reduction we assume that we know $k$ (we can run algorithm for each value of $1 \leq k \leq m$).

\begin{theorem}
\label{hardness_of_FTFL}
If for any $r > 1$ there is a polynomial time algorithm with an approximation factor smaller than $1.278$ for the Fault-Tolerant Facility Location problem for instances with minimal requirement $r$, then $NP \subseteq DTIME[n^{O(log~log~n)}]$.
\end{theorem}

\begin{proof}
 We transform a given instance of the Set Cover $(X, S)$ to an instance of $FTFL$ with uniform requirements $r > 1$. Suppose that we know that an optimal solution to $(X, S)$ uses $k$ sets. Define $R_i$ as a multi-set of singletons $\{x_i\}$ for each $x_i \in X$, ($R_i = \{\{x_i\}, \{x_i\}, \dots, \{x_i\}\}$ and $|R_i| = r - 1$). We set $C = X$ and $F = S \cup R$, where $R = \bigcup_{i = 1}^{n} R_i$. If $x_i \in S_j$ then $c_{ij} = 1$ and $c_{ij} = 3$ otherwise. Similarly if $x_{i'} \in R_{j'}$ then $c_{i'j'} = 0$ and $c_{i'j'} = 2$ otherwise. We extend a distance function $c_{ij}$ by the shortest paths to obtain a metric instance. 
 
 The value of $\gamma$ will be specified later. Suppose $FTFL$-$ALG$ is an $\alpha$ approximation algorithm for the $FTFL$.
\begin{algorithm}
  \caption{$Algorithm~(X, S)$}
\begin{algorithmic}[1]
 \STATE Create FTFL instance $(F, C)$ where $F = S \cup R$ and $C = X$
 \STATE $f_i = \gamma \frac{|C|}{k}$ is cost of a facility $i \in S$ and $f_i = 0$ for $i \in R$
 \WHILE{$C \neq \emptyset$}
  \STATE $F' = FTFL$-$ALG(F, C)$
  \STATE Let $C'$ be a set of clients connected in a distance one with any $f \in F' \setminus R$
  \STATE Let $F = F \setminus (F' \cap S)$ and $C = C \setminus C'$
  \FORALL{$i \in S$}
    \STATE $f_i = \gamma \frac{|C|}{k}$
  \ENDFOR
 \ENDWHILE
\end{algorithmic} 
\end{algorithm}
 Suppose that in iteration $j$ we have $n_j = |C|$ clients not covered by any facility from set $S$. Moreover, the cost of facilities from set $S$ is $\gamma \frac{n_j}{k}$. 
 
 We know that there is a solution to $(X, S)$, which select exactly $k$ sets, such that each element is covered, hence there exists a solution for $(F, C)$ such that each client has an open facility from the set $S$ at distance one and $r-1$ open facilities from the set $R$ at distance 0. The cost of such solution is $n_j(\gamma + 1)$. The approximation ratio of $FTFL$-$ALG$ is $\alpha$, hence on the cost of the returned solution is at most $\alpha n_j(\gamma + 1)$. 
 
 Suppose that $FTFL$-$ALG$ opened $\beta_jk$ facilities from $S$ and all facilities from $R$. Fraction of all clients which are connected in distance one is $0 \leq c_j \leq 1$. The cost of such solution is $\beta_j \gamma n_j + c_jn_j + 2 (n_j - c_jn_j)$. Thus we obtain: $$\alpha \geq \frac{\gamma \cdot \beta_j + 2 - c_j}{1 + \gamma}$$ Define $c_{\beta_j} = (1 - e^{\frac{-\beta_j}{c}})$, where $0 < c < 1$. Suppose that for some iteration $j$, $c_j \leq c_{\beta_j}$ then $$\alpha \geq \frac{\gamma \cdot \beta_j + 1 + e^{\frac{-\beta_j}{c}}}{1 + \gamma}$$ Taking the derivative with respect to $\beta_j$ , tells us that the minimum is achieved at $\beta_j = c \cdot ln(\frac{1}{c\gamma})$. Substituting this value of $\beta_j$ gives us $$\alpha \geq \frac{1 + c\gamma}{1 + \gamma} + \frac{c\gamma}{1 + \gamma}ln(\frac{1}{c\gamma})$$ Choosing $\gamma = 0.278465$ gives $\alpha \geq 1.278465$ when c is a constant close to 1.
 
 In the second case we have that for each $j$, $c_j > c_{\beta_j}$. Similar to the proof in \cite{Guha} it allows us to have an algorithm for the Set Cover with the approximation ratio $c\cdot ln(|X|)$ for $c < 1$, which implies that $NP \subseteq DTIME[n^{O(log~log~n)}]$, see \cite{Feige}. It remains to observe that we may ''guess'' $k$ by trying the above construction for $k = 1, 2, \dots |S|$.
 \qed
\end{proof}

The main idea of the proof, which you can find in the Appendix, is the same as in \cite{Guha}. We use an algorithm for $FTFL$ to obtain partial covers for the Set Cover instance $(X, S)$. We show that the partial cover cannot be too big in each step, because then it would contradict the Feige’s result \cite{Feige}. He proved that approximation algorithm for the Set Cover with ratio $c \cdot ln |X|$ where $c < 1$, implies that $NP \subseteq DTIME[n^{O(log~log~n)}]$.

\begin{figure}
  \centering
  \includegraphics[height=60mm]{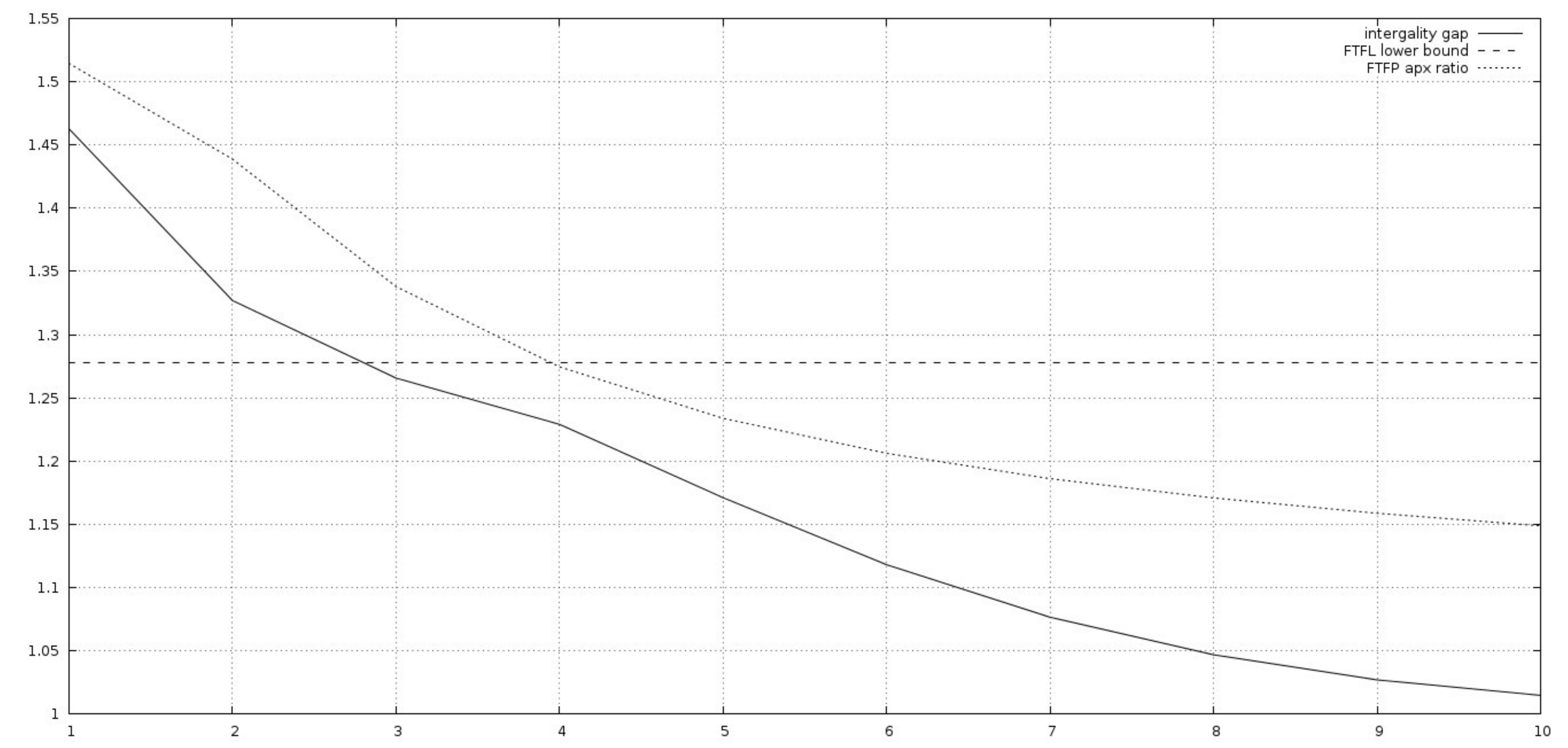}
 \caption{The figure shows three quantities as a function of $r = 1, 2, \dots 10$: the lower bound for $FTFL$, approximation ratio (in the general, non-uniform case) of our algorithm for $FTFP$, and a lower bound on the integrality gap of the LP (\ref{lp_ufl:goal}) - (\ref{lp_ufl:non_negative}). The integrality gap results are also true for the  standard $LP$ for $FTFL$, for details see Appendix~\ref{int_gap_section}.}
  \label{fig:apx_plot}
\end{figure}

\section{Open problems}
Is it possible to apply techniques similar to presented in this paper to $FTFL$? Is $FTFL$ getting any easier with increasing value of $r$? It would also be interesting to derive a lower bound on the approximability of $FTFP$ as a function of $r > 1$.


\begin{thebibliography}{99}
\bibitem{ShiLi} S.~Li: A 1.488 approximation algorithm for the uncapacitated facility location problem. Inf. Comput. 222: 45-58 (2013)
\bibitem{Jain} K.~Jain, M.~Mahdian, A.~Saberi: A new greedy approach for facility location problems. STOC 2002: 731-740
\bibitem{Sviridenko} M.~Sviridenko: An Improved Approximation Algorithm for the Metric Uncapacitated Facility Location Problem. IPCO 2002: 240-257
\bibitem{Aardal} J.~Byrka, K.~Aardal: An Optimal Bifactor Approximation Algorithm for the Metric Uncapacitated Facility Location Problem. SIAM J. Comput. 39(6): 2212-2231 (2010)
\bibitem{Tardos} D.~Shmoys, E.~Tardos, K.~Aardal: Approximation Algorithms for Facility Location Problems (Extended Abstract). STOC 1997: 265-274
\bibitem{liyan_phd} L.~Yan: Approximation Algorithms for the Fault-Tolerant Facility Placement Problem, Phd Thesis
\bibitem{Yan_316} L.~Yan, M.~Chrobak: Approximation algorithms for the Fault-Tolerant Facility Placement problem. Inf. Process. Lett. 111(11): 545-549 (2011)
\bibitem{Feige} U.~Feige: A threshold of ln n for approximating set-cover, 28th ACM Symposium on Theory of Computing, pages 314–318, (1996)
\bibitem{Aravind} R.~Gandhi, S.~Khuller, S.~Parthasarathy, A.~Srinivasan: Dependent rounding and its applications to approximation algorithms. J. ACM 53(3): 324-360 (2006)
\bibitem{ftfl_1725} J.~Byrka, A.~Srinivasan, C.~Swamy: Fault-Tolerant Facility Location: A Randomized Dependent LP-Rounding Algorithm. IPCO 2010: 244-257
\bibitem{Swamy} C.~Swamy, D.~Shmoys: Fault-tolerant facility location. ACM Transactions on Algorithms 4(4) (2008)
\bibitem{Guha} S.~Guha and S.~Khuller. Greedy strikes back: Improved facility location algorithms. In Proceedings of the 9th ACM-SIAM Symposium on Discrete Algorithms (SODA), pages 228–248, SIAM, Philadelphia, 1998
\bibitem{Meyerson} S.~Guha, A.~Meyerson, K.~Munagala: Improved algorithms for fault tolerant facility location. SODA 2001: 636-641
\bibitem{Chudak} F.~Chudak, D.~Shmoys: Improved Approximation Algorithms for the Uncapacitated Facility Location Problem. SIAM J. Comput. 33(1): 1-25 (2003)
\bibitem{Ghodsi} Jaroslaw Byrka, MohammadReza Ghodsi, Aravind Srinivasan: LP-rounding algorithms for facility-location problems. CoRR abs/1007.3611 (2010)
\bibitem{Yan} Li Yan, Marek Chrobak: LP-rounding Algorithms for the Fault-Tolerant Facility Placement Problem. CoRR abs/1205.1281 (2012)
\bibitem{Xu} S.~Xu, H.~Shen: The Fault-Tolerant Facility Allocation Problem. ISAAC 2009: 689-698
\end{thebibliography}
\end{document}